\documentclass{amsart}
\usepackage{amsfonts, amstext, amsmath, amssymb}
\usepackage[all]{xy}

\newcommand{\C}{\mathbb{C}}

\begin{document}

\newtheorem{theorem}{Theorem}
\newtheorem{proposition}{Proposition}

\bibliographystyle{plain}

\title{Frames for Subspaces of $\C^N$}
\author[M. Hirn, D. Widemann]{Matthew Hirn, David Widemann}
\address{Norbert Wiener Center, University of Maryland, College Park, MD 20742}
\date{November 19, 2007}
\keywords{Frames, frame potential}
\begin{abstract}
We present a theory of finite frames for subspaces of $\C^N$. The definition of a subspace frame is given and results analogous to those from frame theory for $\C^N$ are proven.
\end{abstract}
\maketitle

\section{Introduction}

Frames have been used in many applications of signal processing. They give stable signal representations and allow modelling for noisy environments. Recent work with hyperspectral data has shown the need to consider subspaces of high dimensional spaces. In particular, ongoing work has led us to investigate frames for these subspaces.

\section{Frames for $\C^N$}

Let $\Phi = \{\varphi_j\}_{j=1}^s \subset \C^N$, where $s \geq N$. $\Phi$ is a {\it finite frame} for $\C^N$ if there exist constants $A,B > 0$ such that
\begin{equation}\label{eqn:frame}
A\lVert f \rVert^2 \leq \sum_{j=1}^s |\langle f,\varphi_j\rangle |^2 \leq B\lVert f\rVert^2, \qquad \forall \enspace f \in \C^N.
\end{equation}
The numbers $A,B$ are called the {\it frame bounds}. It is a well known fact that any spanning set is a frame for $\C^N$, while every frame is indeed a spanning set. A frame is {\it tight} if one can choose $A = B$ in the definition, i.e., if
\begin{equation}\label{eqn:tight}
\sum_{j=1}^s |\langle f,\varphi_j\rangle |^2 = A\lVert f\rVert^2, \qquad \forall \enspace f \in \C^N.
\end{equation}
Finally, a frame is {\it unit norm} if
\begin{equation}\label{eqn:unit norm}
\lVert \varphi_j \rVert = 1, \qquad \forall \enspace j=1,\ldots ,s.
\end{equation}
If $\Phi$ satisfies (\ref{eqn:frame}), (\ref{eqn:tight}), and (\ref{eqn:unit norm}), then we say $\Phi$ is a {\it finite unit norm tight frame} (FUNTF) for $\C^N$. In this case, the frame bounds satisfy $A = B = s/N$. In particular, if $\Phi$ is a FUNTF with frame bounds $A = B = 1$, then $\Phi$ is an orthonormal basis.\\

Assume now that $\Phi = \{\varphi_j\}_{j=1}^s$ is a frame for $\C^N$. The {\it analysis operator} of $\Phi$ is defined as follows:
\begin{equation*}\label{eqn:adjoint}
L:\C^N \rightarrow \C^s, \quad Lf := \{\langle f,\varphi_j\rangle\}_{j=1}^s.
\end{equation*}
The {\it synthesis operator} is given by:
\begin{equation*}\label{eqn:synthesis}
L^{\star}:\C^s \rightarrow \C^N, \quad L^{\star}\{c_j\}_{j=1}^s = \sum_{j=1}^s c_j\varphi_j.
\end{equation*}
One obtains the {\it frame operator} by composing $L^{\star}$ with $L$:
\begin{equation*}\label{eqn:frame operator}
S:\C^N \rightarrow \C^N, \quad Sf = L^{\star}Lf = \sum_{j=1}^s \langle f,\varphi_j\rangle\varphi_j
\end{equation*}
Some important properties of $S$ are the following:
\begin{enumerate}
\renewcommand{\labelenumi}{(\roman{enumi})}
\item
$S$ is invertible and self-adjoint.
\item
Every $f \in \C^N$ can be represented as
\begin{equation}\label{eqn:rep formula}
f = \sum_{j=1}^s \langle f,S^{-1}\varphi_j\rangle\varphi_j = \sum_{j=1}^s \langle f,\varphi_j\rangle S^{-1}\varphi_j.
\end{equation}
\item
$\Phi$ is a tight frame if and only if $S = AI$.
\end{enumerate}
Based on equation (\ref{eqn:rep formula}), one defines the {\it dual frame} of $\Phi$ as $\tilde{\Phi} = \{\tilde{\varphi}_j\}_{j=1}^s := \{S^{-1}\varphi_j\}_{j=1}^s$; the frame operator of $\tilde{\Phi}$ is $S^{-1}$. If $\Phi$ is a tight frame for $\C^N$, then $S^{-1} = \frac{1}{A}I$, and the representation formula is simple:
\begin{equation*}\label{eqn:rep formula funtf}
f = \frac{1}{A}\sum_{j=1}^s \langle f,\varphi_j\rangle\varphi_j = \frac{1}{A}\sum_{j=1}^s \langle f,\varphi_j\rangle\varphi_j.
\end{equation*}

\section{Frames for Subspaces of $\C^N$}

Let $\Phi = \{\varphi_j\}_{j=1}^s \subset \C^N$ and let $W$ be a subspace of $\C^N$ of dimension $r < N$. We say $\Phi$ is a {\it finite subspace frame} for $W$ if $\mathrm{span}(\Phi) = W$. It is clear from this definition that there exist constants $A,B > 0$ such that
\begin{equation}\label{eqn:subframe}
A\lVert f \rVert^2 \leq \sum_{j=1}^s |\langle f,\varphi_j\rangle |^2 \leq B\lVert f\rVert^2, \qquad \forall \enspace f \in W.
\end{equation}
We note that if we had instead used (\ref{eqn:subframe}) as our definition, then it would not necessarily imply that $\mathrm{span}(\Phi) = W$ but rather that $\mathrm{span}(\Phi) \supseteq W$. The unit norm property as well as the notion of a tight frame remain similar in this setting. More specifically, if we can take $A = B$ in (\ref{eqn:subframe}) then we call $\Phi$ a {\it tight subspace frame}. Finally, if $\Phi$ is a finite unit norm tight subspace frame, then we say $\Phi$ is a {\it subspace FUNTF}.\\

We define $L$, $L^{\star}$, and $S$ exactly the same as in section 1, however we note that the properties of these maps change for subspace frames. In particular, we see:
\begin{enumerate}
\renewcommand{\labelenumi}{(\alph{enumi})}
\item
$L:\C^N \rightarrow \C^s$ is no longer injective, but rather $\mathrm{ker}(L) = (\C^N\setminus W) \cup \{0\}$.
\item
$L^{\star}:\C^s \rightarrow \C^N$ is no longer surjective, but rather $\mathrm{image}(L^{\star}) = W$.
\item
Based on (a) and (b), we see that $S:\C^N \rightarrow \C^N$ is no longer invertible.
\end{enumerate}
Because of (c), none of properties (i) - (iii) from section 1 hold for subspace frames. The question then becomes: in what sense do subspace frames satisfy properties (i) - (iii) above? Theorems below show that subspace frames satisfy natural modifications of the above properties.\\

Let $W_{on}$ be a set of $r$ orthonormal vectors such that $\mathrm{span}(W_{on}) = W$. We will also consider $W_{on}$ as an $N \times r$ matrix where the columns of this matrix are the vectors in the set $W_{on}$. We define $\Phi_W$ to be the $r \times s$ matrix whose columns are the coordinates of $\Phi$ in $W_{on}$; that is:
\begin{equation}\label{eqn:coordsubframe}
\Phi_W := W_{on}^{\star}\Phi,
\end{equation}
where we have implicitly used the matrix form of $\Phi$, that is the $N \times s$ matrix whose columns are the elements of $\Phi$. The $j^{\text{th}}$ column of $\Phi_W$ is the projected $W$-subspace coordinates of $\Phi$.
\begin{proposition}\label{prop:coordframe}
The set $\Phi_W$ consisting of the columns of the matrix $\Phi_W$ is a frame for $\C^r$.
\end{proposition}
\begin{proof}
Since $\mathrm{span}(W_{on}) = W$, we have $\mathrm{ker}(W_{on}^{\star}) \cap W = \{0\}$. Therefore, since $\mathrm{span}(\Phi) = W$ as well, we see that $W_{on}^{\star}\Phi$ has rank $r$.
\end{proof}
We denote the analysis, synthesis, and frame operators of $\Phi_W$ by $L_W$, $L_W^{\star}$, and $S_W$, respectively. In terms of the analysis operator, $L$, for $\Phi$, $L_W = LW_{on}$. By proposition \ref{prop:coordframe} we see that $S_W$ will satisfy (i) - (iii).
\begin{theorem}\label{thm:subfuntf equiv}
$\Phi$ is a subspace FUNTF for $W$ with frame bound $A$ if and only if $\Phi_W$ is a FUNTF for $\C^r$ with frame bound $A$.
\end{theorem}
\begin{proof}
We do the forward direction first: let $g \in \C^r$, then:
\begin{eqnarray*}
\langle S_Wg,g\rangle & = & \langle L_Wg,L_Wg\rangle\\
& = & \langle \Phi^{\star}W_{on}g,\Phi^{\star}W_{on}g\rangle\\
& = & \sum_{j=1}^s |\langle W_{on}g,\varphi_j\rangle|^2\\
& = & A\lVert W_{on}g\rVert^2\\
& = & A\langle W_{on}g,W_{on}g\rangle
\end{eqnarray*}
Therefore we have:
\begin{eqnarray*}
\langle S_Wg,g\rangle - A\langle W_{on}g,W_{on}g\rangle & = & 0 \qquad \Longrightarrow\\
\langle S_Wg,g\rangle - A\langle W_{on}^{\star}W_{on}g,g\rangle & = & 0 \qquad \Longrightarrow\\
\langle g,(S_W - AI)g\rangle & = & 0 \qquad \Longrightarrow\\
S_W & = & AI
\end{eqnarray*}
For the reverse direction, let $f \in W$. There exists $g \in \C^r$ such that $W_{on}g = f$. Therefore,
\begin{eqnarray*}
A\lVert f \rVert^2 & = & A \langle f,f\rangle \\
& = & A\langle W_{on}g, W_{on}g\rangle \\
& = & \langle Ag, g \rangle \\
& = & \langle S_Wg,g \rangle \\
& = & \langle W_{on}^{\star}L^{\star}LW_{on}g,g \rangle \\
& = & \langle LW_{on}g,LW_{on}g \rangle \\ 
& = & \langle Lf, Lf \rangle \\
& = & \sum_{j= 1}^{s}|\langle f,\varphi_j \rangle|^2
\end{eqnarray*}
\end{proof}
We define the dual frame of $\Phi_W$ in the usual way, that is $\tilde{\Phi}_W = S_W^{-1}\Phi_W$. We now define the {\it dual subspace frame} of $\Phi$ as follows:
\begin{equation}\label{eqn:dual subframe}
\tilde{\Phi} := W_{on}\tilde{\Phi}_W = W_{on}S_W^{-1}W_{on}^{\star}\Phi.
\end{equation}
As the name implies, the set $\tilde{\Phi} = \{\tilde{\varphi}_j\}_{j=1}^s = \{W_{on}S_W^{-1}W_{on}^{\star}\varphi_j\}_{j=1}^s$ will have the following properties:
\begin{proposition}\label{prop:subframe}
$\tilde{\Phi}$ is a subspace frame for $W$.
\end{proposition}
\begin{proof}
This follows from proposition \ref{prop:coordframe}.
\end{proof}
\begin{theorem}\label{thm:dual subframe}
Every $f \in \mathrm{W}$ can be represented as
\begin{equation*}
f = \sum_{j=1}^s \langle f,\tilde{\varphi}_j\rangle\varphi_j = \sum_{j=1}^s \langle f,\varphi_j\rangle \tilde{\varphi}_j.
\end{equation*}
\end{theorem}
\begin{proof}
The first representation formula is $\Phi\tilde{\Phi}^{\star}f = f$ for all $f \in W$. Letting $f = W_{on}g$ for some $g \in \C^r$, we have:
\begin{eqnarray}
\Phi\tilde{\Phi}^{\star}f & = & \Phi(W_{on}S_W^{-1}W_{on}^{\star}\Phi)^{\star}f \nonumber\\
& = & \Phi\Phi^{\star}W_{on}(S_W^{-1})^{\star}W_{on}^{\star}(W_{on}g) \nonumber\\
& = & SW_{on}S_W^{-1}g \nonumber\\
& = & SW_{on}(W_{on}^{\star}SW_{on})^{-1}g \label{pfln:repform1}
\end{eqnarray}
Since $W_{on}W_{on}^{\star}$ is the identity on $W$,
\begin{eqnarray*}
(\ref{pfln:repform1}) & = & W_{on}W_{on}^{\star}SW_{on}(W_{on}^{\star}SW_{on})^{-1}g\\
& = & W_{on}g\\
& = & f
\end{eqnarray*}
The second representation formula is $\tilde{\Phi}\Phi^{\star}f = f$ for all $f \in W$.
\begin{eqnarray*}
\tilde{\Phi}\Phi^{\star}f & = & (W_{on}S_W^{-1}W_{on}^{\star}\Phi)\Phi^{\star}f\\
& = & W_{on}(W_{on}^{\star}SW_{on})^{-1}W_{on}^{\star}SW_{on}g\\
& = & W_{on}g\\
& = & f
\end{eqnarray*}
\end{proof}
The following commutative diagram illustrates the above ideas:\\\\
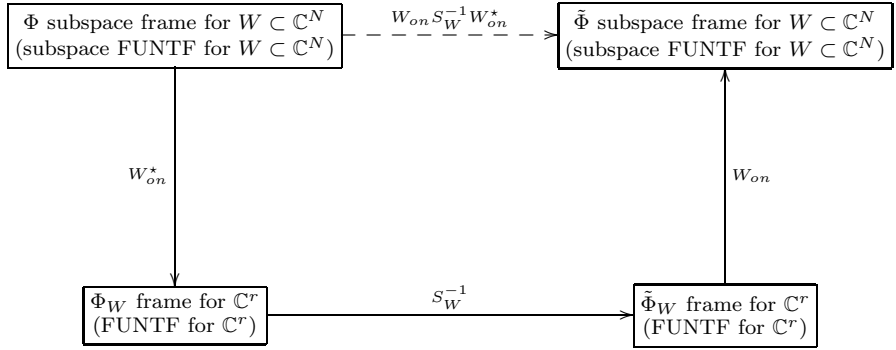
\begin{displaymath}
\footnotesize
\xymatrix{
*+[F]{\txt{$\Phi$ subspace frame for  $W \subset \C^N$\\ (subspace FUNTF for $W \subset \C^N$)}} \ar@{-->}[rrr]^{W_{on}S_W^{-1}W_{on}^{\star}} \ar[ddd]_{W_{on}^{\star}} &&& *+[F]{\txt{$\tilde{\Phi}$ subspace frame for  $W \subset \C^N$\\ (subspace FUNTF for $W \subset \C^N$)}}\\\\\\
*+[F]{\txt{$\Phi_W$ frame for $\C^r$\\ (FUNTF for $\C^r$)}} \ar[rrr]^{S_W^{-1}} &&& *+[F]{\txt{$\tilde{\Phi}_W$ frame for $\C^r$\\ (FUNTF for $\C^r$)}} \ar[uuu]_{W_{on}}
}
\end{displaymath}
\\

\section{Frame Potential}

Define the frame potential of a finite unit norm frame $\Phi = \{\varphi_j\}_{j=1}^s$ for $\C^N$ as:
\begin{equation*}
\mathrm{FP}(\Phi) := \sum_{j=1}^s \sum_{k=1}^s |\langle \varphi_j,\varphi_k \rangle|^2.
\end{equation*}

In \cite{benedetto:fntf} a characterization of FUNTFs is given in terms of the frame potential:
\begin{theorem}[Benedetto and Fickus 2002]\label{thm:fp}
For a given $N$ and $s$, let $S^{N-1}$ denote the unit sphere in $\C^N$ and consider:
\begin{displaymath}
\mathrm{FP}:\underbrace{S^{N-1} \times \cdots \times S^{N-1}}_{s \text{ times}} \rightarrow [0,\infty).
\end{displaymath}
Then:
\begin{enumerate}
\item
Every local minimizer of the frame potential is also a global minimizer.
\item
If $s\leq N$, the minimum value of the frame potential is $s$, and the minimizers are precisely the orthonormal sequences in $\C^N$.
\item
If $s \geq N$, the minimum value of the frame potential is $s^2/N$, and the minimizer are precisely the FUNTFs for $\C^N$.
\end{enumerate}
\end{theorem}

The following theorem is a trivial generalization of theorem \ref{thm:fp}:\\

\begin{theorem}\label{thm:sfp}
For a given $s$ and $N$, let $W$ be a subspace of $\C^N$ of dimension $r < N$ and consider the resctricted frame potential:
\begin{displaymath}
\mathrm{FP}|_W:(\underbrace{S^{N-1} \times \cdots \times S^{N-1}}_{s \text{ times}}) \cap W \rightarrow [0,\infty).
\end{displaymath}
Then:
\begin{enumerate}
\item
Every local minimizer of the restricted frame potential is also a global minimizer.
\item
If $s\leq r$, the minimum value of the restricted frame potential is $s$, and the minimizers are precisely the orthonormal sequences in $W$.
\item
If $s \geq r$, the minimum value of the restricted frame potential is $s^2/r$, and the minimizer are precisely the subspace FUNTFs for $W$.
\end{enumerate}
\end{theorem}

Theorem \ref{thm:sfp} shows that the minimum value of the frame potential depends on the dimension of the subspace $W$.
\begin{proof}
Let $W_{on}$ be a set of $r$ orthonormal vectors such that $\mathrm{span}(W_{on}) = W$ and consider it as an $N \times r$ matrix. If $\Phi = \{\varphi_j\}_{j=1}^s$ is a finite unit norm set of vectors in $W$, then the coordinates of $\Phi$ in $W_{on}$ are given by the $r \times s$ matrix $\Phi_W = W_{on}^{\star}\Phi$. In \cite{benedetto:fntf} it is shown that $\mathrm{FP}(\Phi) = \mathrm{Tr}(S^2)$, where $S$ is the frame operator of $\Phi$. Using the previous two statements we then have:
\begin{eqnarray*}
\mathrm{FP}|_W(\Phi) & = & \mathrm{Tr}(S^2)\\
& = & \mathrm{Tr}([(W_{on}\Phi_W)(W_{on}\Phi_W)^{\star}]^2)\\
& = & \mathrm{Tr}([\Phi_W\Phi_W^{\star}]^2)\\
& = & \mathrm{Tr}(S_W^2)\\
& = & \mathrm{FP}(\Phi_W)
\end{eqnarray*}

Since $\Phi_W$ is a finite unit norm set of vectors in $\C^r$, we can apply theorem \ref{thm:fp} to get (1) and (2). Combining theorem \ref{thm:fp} along with theorem \ref{thm:subfuntf equiv} gives (3).
\end{proof}
\bigskip

\nocite{benedetto:fntf}
\nocite{kovacevic:lbbaf}

\bibliography{sensipBib3}

\end{document}